\documentclass[onecolumn]{IEEEtran}
\IEEEoverridecommandlockouts

\usepackage[lmargin=0.75in,rmargin=0.75in,tmargin=.75in,bmargin=0.7in]{geometry}

\usepackage{amsmath}
\usepackage{amsfonts}
\usepackage{amssymb}
\usepackage{cite}
\usepackage{calc}
\usepackage{color}
\usepackage{epsfig}
\usepackage{hyperref}
\usepackage{enumerate}
\usepackage{bm}
 \usepackage{caption}
 \usepackage{graphicx}
 \usepackage{subcaption}
 \usepackage{breqn}

\newtheorem{defn}{Definition}
\newtheorem{thm}{Theorem}[section]
\newtheorem{cor}[thm]{Corollary}
\newtheorem{prop}{Proposition}
\newtheorem{lem}[thm]{Lemma}
\newtheorem{example}{Example}
\newtheorem{conj}[thm]{Conjecture}
\newtheorem{constr}[thm]{Construction}
\newtheorem{note}{Remark}

\newcommand{\bit}{\begin{itemize}}
\newcommand{\eit}{\end{itemize}}
\newcommand{\bcor}{\begin{cor}}
\newcommand{\ecor}{\end{cor}}
\newcommand{\beq}{\begin{equation}}
\newcommand{\eeq}{\end{equation}}
\newcommand{\beqn}{\begin{equation*}}
\newcommand{\eeqn}{\end{equation*}}
\newcommand{\bea}{\begin{eqnarray}}
\newcommand{\eea}{\end{eqnarray}}
\newcommand{\bean}{\begin{eqnarray*}}
\newcommand{\eean}{\end{eqnarray*}}
\newcommand{\ben}{\begin{enumerate}}
\newcommand{\een}{\end{enumerate}}
\newcommand{\bdefn}{\begin{defn}}
\newcommand{\edefn}{\end{defn}}
\newcommand{\bnote}{\begin{note}}
\newcommand{\enote}{\end{note}}
\newcommand{\bprop}{\begin{prop}}
\newcommand{\eprop}{\end{prop}}
\newcommand{\blem}{\begin{lem}}
\newcommand{\elem}{\end{lem}}
\newcommand{\bthm}{\begin{thm}}
\newcommand{\ethm}{\end{thm}}
\newcommand{\bconj}{\begin{conj}}
\newcommand{\econj}{\end{conj}}
\newcommand{\bconstr}{\begin{constr}}
\newcommand{\econstr}{\end{constr}}
\newcommand*{\Resize}[2]{\resizebox{#1}{!}{$#2$}}%

\title{Optimal Haplotype Assembly from High-Throughput Mate-Pair Reads}

\author{
Govinda M. Kamath$^{1}$, Eren \c{S}a\c{s}o\u{g}lu$^{2}$ and David Tse$^{1}$ \\
$^{1}$ Department of Electrical Engineering, Stanford University, Stanford, USA.\\
$^{2}$ Department of Electrical Engineering and Computer Science, University of California, Berkeley,  USA.\\
Email: gkamath@stanford.edu, eren@eecs.berkeley.edu, dntse@stanford.edu.\thanks{ \scriptsize This work is partially supported by the Center for Science of Information (CSoI), an NSF
Science and Technology Center, under grant agreement CCF-0939370.}}
\date{\today}

\begin{document} \maketitle
 \begin{abstract}
 Humans have $23$ pairs of homologous chromosomes. The homologous pairs are almost identical pairs of chromosomes.
 For the most part, differences in homologous chromosome occur at certain documented positions called single nucleotide polymorphisms (SNPs).
 A haplotype of an individual is the pair of sequences of SNPs on the two homologous chromosomes. In this paper, we study the problem of inferring haplotypes of individuals from mate-pair reads of their genome. We give a simple formula for the coverage needed for haplotype assembly, under a  generative model. The analysis here leverages connections of this problem with decoding convolutional codes.
\end{abstract}
 \section{Introduction}
 
 Humans, like most mammals are diploid organisms, $i.e.$ all somatic
 cells of humans contain two copies of the genetic material. Humans
 have $23$ pairs of homologous chromosomes. Each member of a pair of
 homologous chromosomes, one paternal and the other maternal,  has
 essentially the same genetic material, one exception being the sex-chromosomes in males. For the most part, homologous
 chromosomes differ in their bases (which take values in the
 alphabet $\{A,C,G,T\}$) at specific positions known as
 Single Nucleotide Polymorphisms (SNP). Other types of differences
 such as insertions and deletions are also possible, but these are
 rare occurrences, and are ignored in this paper. There are
 around $3$ million known SNPs in humans, whose genome is of length
 approximately $3$ billion base pairs. Thus on average a SNP appears
 once in  $1000$ base pairs. The positions where SNPs occur are well
 documented (see for instance \cite{dbSNP}). Usually only one of two bases
 can be seen at any particular SNP position in the population. The one that is seen in the
 majority of a population is referred to as the {\it  major allele}, and the
 other is referred to as the {\it minor allele}. The two sequences of SNPs,
 one on each of the homologous chromosomes, is called the {\it haplotype}
 of an individual.  A person is said to be {\it homozygous} at
 position $i$ if both homologous chromosomes have the same base at
 position $i$.  Otherwise the person is said to be {\it heterozygous}
 at position $i$.
 
  The haplotype of an individual provides important information in applications
  like personalized medicine, and understanding phylogenetic trees.
  The standard method to find a person's haplotype is to first find
  her \emph{genotype}, which is the set of allele pairs in her genome.
  Haplotype information, i.e., which allele lies in which chromosome, is then statistically inferred from a set of
  previously known haplotypes based on population genetics models \cite{BroBro}.  Often, haplotypes of
  many individuals from a population are inferred jointly. 
  This process is referred to as {\it
  haplotype phasing}. A major drawback of this approach is that many
  individuals in a population need to be sequenced to get reliable estimates
  of the haplotype of one person.
 
 The advent of next generation sequencing technologies provides an
 affordable alternative
 to haplotype phasing.  In particular, these technologies allow one to
 quickly and cheaply read the bases of hundreds of millions of short genome fragments, called
 \emph{reads}. 
 One can
 \emph{align} such reads to a known human genome, thereby
 determining their physical location in the chromosome.  
 Aligning a read to a reference does not reveal whether that read comes from the paternal or
 the maternal chromosome, however, hence the non-triviality of determining the
 haplotype.  

 Clearly, if a read
 covers zero or one SNP, it does not contain any information about how
 a particular SNP relates to other SNPs on the same chromosome, and
 is thus useless for haplotype assembly.  That is, a read helps in determining the haplotype only 
 if it covers at least two SNPs.  This may seem like a problem at
 first sight, since as we mentioned above, adjacent SNPs are separated on
 average by $1,000$ bases, but most
 sequencing technologies produce reads of length only a few
 hundred bases.  Fortunately, with somewhat more labour intensive library preparation 
 techniques, one can produce {\it mate-pair reads}, i.e., reads that 
 consist of two genome fragments separated by a number of bases. The
 length of the separation between the two segments of DNA read is
 known as the {\it insert size}. For popular technologies like
 Illumina, the read lengths are around $90-100$ base pairs (bp),
 and the insert size ranges from around $300$ bp to $10,000$ bp,
 with the median insert size around $3,000$ bp (See Figure \ref{fig:insert_size}, which was
 taken from  \cite{Nextera}). These
 reads offer a possibility of inferring haplotypes from them. However,
 errors in reads present a major challenge. Both
 single and mate-pair reads are illustrated in Figure \ref{fig:read_generation}. 
 
     \begin{figure}
  \centering
   \includegraphics{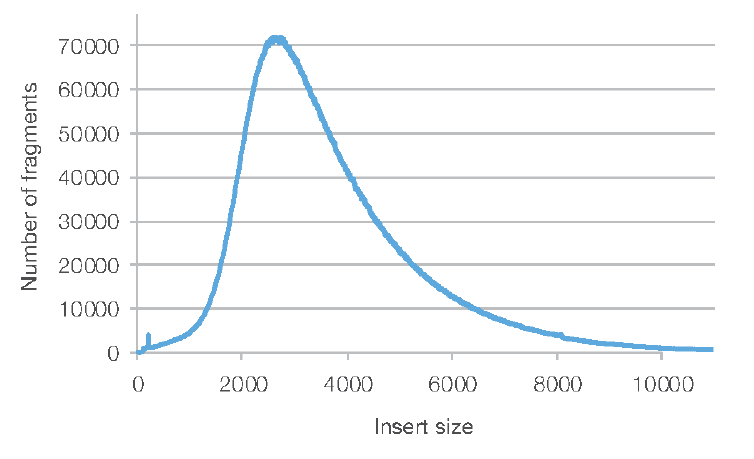}
   \caption{\small Insert size distribution of Human NA $12877$ library sequenced on Illumina HiSeq\textsuperscript{\circledR} $2500$. Median insert size was $3400$ bp here. Figure taken from \cite{Nextera}. }
   \label{fig:insert_size}
 \end{figure}

 In this paper, we characterize the coverage, i.e.,
 the number of reads that cover at least two SNP positions, 
 required for haplotype assembly from paired reads.  We determine this
 quantity as a function of the number of SNPs $n$, 
 read error probability $p$, 
 and the separation between the two reads in each  mate-pair reads in terms of number of SNPs given by a random variable $W$. In particular, we
 show that for uniformly distributed read locations, the coverage
 required for a maximum-likelihood (ML) assembler to succeed is asymptotically
 $$
 \frac{n\log n}{\min\{\mathbb{E}[W],2\}(1-e^{-D(0.5||2p(1-p))})},
 $$  
 and no assembler can succeed on all possible haplotypes with fewer reads. 
 In proving this result, we show the haplotype assembly problem's connection to
 that of decoding a convolutional code, and use analysis methods
 for such codes. We note that we refer to natural logarithms by $\log$ throughout the manuscripts.
 
      \begin{figure}
  \centering
   \includegraphics[ clip=true, scale=0.6]{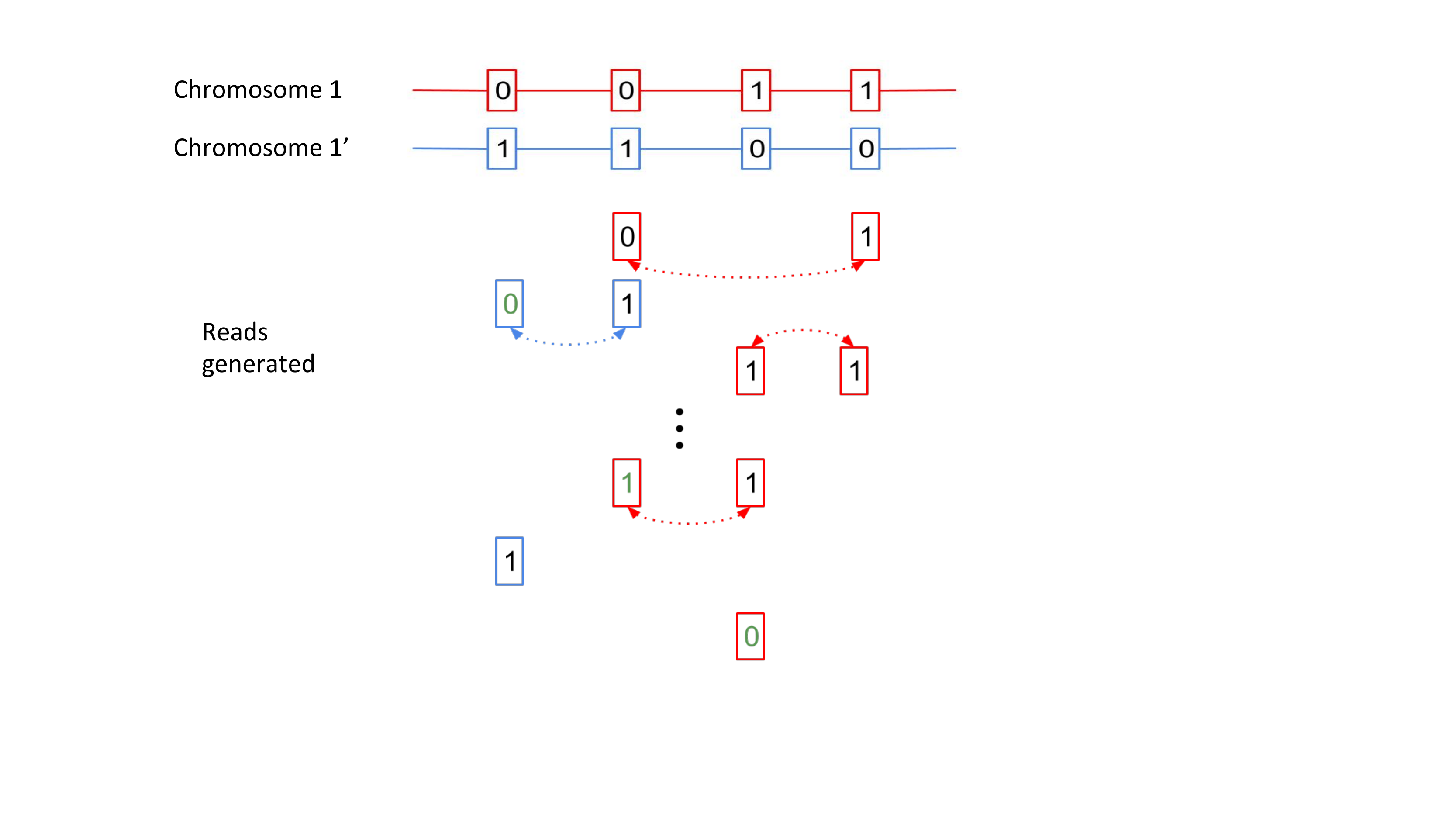}
   \caption{\small An illustration of read generation. The genetic material in chromosome $1$ and chromosome $2$ are identical apart from in the SNP positions shown. Here we represent the major allele by $0$ and the minor allele by $1$. Reads from chromosome $1$ are represented in red and those from chromosome $2$ are represented in red. We do not know which reads come from which chromosome. Mate-pair reads are indicated by dotted lines in between them. Errors are shown in green. \vspace{0.1 in} }
   \label{fig:read_generation}
 \end{figure}

 The haplotype assembly problem has been studied previously.
 In particular, that maximum-likelihood (ML) estimation can be done
 using dynamic programming 
 in a similar setting was first recognized in 
 \cite{HeEtal10}.   
 The problem was also studied recently in \cite{SiVikVis}.
 There, it is assumed that a mate-pair
 could cover any two SNPs on the chromosome and bounds are derived on the
 number of reads needed to carry out haplotype assembly using a
 message passing algorithm. This however is an unrealistic assumption,
 as the chromosome lengths are several million, but the insert lengths
 are typically in the thousands. In this paper, we consider a
 model where reads only cover SNPs which are nearby.

   \section{Noiseless reads covering adjacent SNPs}  
  Consider first the case where all reads are noiseless. 
  For convenience, we shall refer to the
  major allele by $1$ and the minor allele by $0$.  We assume that the
  genome is heterozygous at all SNP positions, an assumption we will justify 
  shortly.  That is, if $S_1,S_2,
  \cdots,S_n$ are the $n$ SNPs on a chromosome, then the other
  chromosome has SNPs $S_1^c,S_2^c, \cdots,
  S_n^c$.  Each mate-pair covers adjacent SNPs
  $i$ and $i+1$, where $i$ is picked from $\{1,2, \cdots, n-1 \}$
  uniformly at random.  As mentioned above, the reads are assumed to
  be aligned  to the reference genome,
  and therefore the genomic locations of all reads are known. A
  read covering SNP $i$ and SNP $i+1$, will output $(S_i, S_{i+1})$
  with probability $\frac{1}{2}$ and SNP  $(S_i^c, S_{i+1}^c)$ with
  probability $\frac{1}{2}$ depending on which of the two homologous
  chromosomes it is read from.    

 Determining the necessary and sufficient coverage in this setting is
 relatively simple.  Indeed, note that the parity of $S_i$ and $S_{i+1}$, which
 equals the parity of $S_i^c$ and $S_{i+1}^c$ is a sufficient
 statistic we get from reads covering SNP $i$ and SNP $i+1$. Thus we
 can reconstruct the haplotype when we see a read covering each of the
 $n-1$ adjacent pair of SNPs.  Note that this is equivalent to the coupon collector
 problem, and therefore reconstruction will be correct with high
 probability if one has 
 $c(n-1)\log (n-1)$ reads with $c>1$, whereas if $c<1$, then reconstruction
 will fail with high probability.

  \section{Noisy reads covering adjacent SNPs}\label{sec:adjacent}
  Next, consider the case where each mate-pair again covers a pair of
  adjacent SNPs and aligned perfectly to a reference genome, but
  now each read in each pair is in error at the SNP position
  with probability $p$ independent of everything else.  Similarly to the noiseless case, we wish to
  characterize optimal coverage as a function of $p$.  We will see
  that $O(n\log n)$ is the correct scaling here as well. 
  In particular, we will say that $c(p)$ is the optimal coverage if 
 for any $\epsilon>0$, 
 $(c(p)-\epsilon)n\log n$ reads are insufficient to reconstruct with
 probability at least $1-\epsilon$, but
 $(c(p)+\epsilon)n\log n$ are sufficient as $n \rightarrow \infty$. This makes optimal
 coverage only a function of $p$.
  
  We again assume that each SNP position is heterozygous.  This is a
  reasonable 
  assumption since the positions of SNPs in
  humans are typically known, and thus one can test every SNP position for
  heterozygosity.
  This can be done using reads that cover a
  single SNP position, as we mentioned above, such reads are
  much easier to obtain in large numbers compared with mate-pairs.
  Thus, the coverage required for reliable heterozygosity testing can be met
  easily.  
 
  We will not fix the number of reads $M$ to be $c(p)(n-1)\ln n$ but instead allow $M$ to be random,
 and
 in particular, be distributed as $\text{Poiss}(c(p)(n-1)\ln n)$.  This relaxation
 only simplifies the analysis and does not affect the result.  Indeed,
 note that such a random variable is the sum of $n\ln n$ i.i.d.
 $\text{Poiss}(c(p))$ random variables, and therefore by the law of large
 numbers will take values in $[c(p)\pm\epsilon]n\ln n$ with high
 probability.
 We assume that SNP positions $1,2,\cdots, n-1$ are
 equally likely to be the first position covered by each read.

 Here again, the set of parities of adjacent SNPs $i$ and $i+1$
 are a sufficient statistic to reconstruct.  
 As the probability of error in each SNP in each read is $p$, and
 errors are independent
 across reads, we have that the probability that a read gives us the wrong parity is 
 \begin{equation}
  \theta=2p(1-p).
 \end{equation}

 Let $L_1, L_2, \cdots, L_{n-1}$ denote the true parities, $i.e.$
 $L_i=S_i+S_{i+1}=S_i^c+S_{i+1}^c$. For example, in Figure \ref{fig:read_generation}, $n=4$, $S_1=S_2=0$, $S_3=S_4=1$, $L_1=0=L_3$, $L_2=1$.
 For $\theta \in [0,1]$, let the function $D(\theta)$ be defined as,
 \begin{equation}
  D(\theta)=D(0.5||\theta),
 \end{equation}
  where $D(\cdot||\cdot)$ denotes relative entropy measured in nats.  Also let 
  $Z_{i1},\cdots, Z_{i,N_i}$ be the observations of $L_i$, and note that due to
  Poisson splitting, we have $N_i \sim$ Poiss$(c(p)\log n)$. Then we have that,
  \begin{equation}
   Z_{i,j}= L_i \oplus \epsilon_{i,j},
  \end{equation}
  where $\epsilon_{i,j}$  are Bern$(\theta)$ random variables
  independent of each other, and $\oplus$ denotes modulo-$2$ addition.
  The assembly problem therefore is equivalent to decoding the $L_i$s from the observations
  $Z_{i,j}$s, as in a repetition code, where each message symbol ($L_i$ $1 \leq i \leq n-1$) is repeated $N_i$ times.

%
 
 
 Under these assumptions, we see that the maximum-likelihood (ML) rule is to declare
 $L_i=0$ (i.e., SNPs $(i,i+1)$ agree) if more than $\frac{N_i}{2}$
 of the observations are $0$ and declare $L_i=1$ (SNPs $(i,i+1)$
 disagree) otherwise.  
 Note that each observation is correct with probability $\theta$, and thus the number of correct observations $\sim$ Bin$(N_i,\theta)$.  Thus we have that,
 \begin{eqnarray}
  \Resize{4.5 cm}{P(\text{wrong parity for } (i,i+1) | N_i)} &\leq&  \Resize{3.3 cm}{P( \text{Bin}(N_i,\theta) \geq \frac{N_i}{2} | N_i)}\nonumber \\ 
  &\leq& e^{-N_i D(\theta)}.
 \end{eqnarray}

 Assuming that $\theta < \frac{1}{2}$, we also have that (from chapter $12$ of \cite{CovTho}) 
 \begin{eqnarray}
  P\left(
  \text{\small wrong parity}
  \text{ for } (i,i+1)
 \big| N_i \right) &\geq& \frac{1}{2} P \left( \text{Bin}(N_i,\theta) \leq \left\lfloor \frac{N_i}{2} \right\rfloor | N_i\right), \nonumber \\
  &\geq& \frac{1}{2} \frac{1}{(N_i+1)^2} e^{-N_i D(\theta)},
 \end{eqnarray}
 where we have used the fact that $D(p||q)$ is a monotonically increasing of $p$ for a fixed $q$, in the regime $p>q$.
 Some arithmetic then reveals that 
 \begin{eqnarray}
  P(\text{wrong parity for } (i,i+1)) 
  \le \left(\frac{1}{n}\right)^{c(p)(1-e^{-D(\theta)})},
 \end{eqnarray}
 
 and that
 \begin{align}
  P \left(
  \text{\small wrong parity}
  \text{ for } (i,i+1)
 \right)  
  &\ge  \frac{1}{2} (c(p) e^{D(\theta) }\log n )^{-2} \Bigg(  \left(\frac{1}{n}\right)^{c(p) (1-e^{-D(\theta)})}   - \frac{1}{n} - \left(\frac{1}{n}  \right)c(p)e^{D(\theta)} \log n   \Bigg) \label{eq:converse1}
 \end{align}
 
 This leads to the following coverage result:
 \begin{thm}
  \[c(p)=\frac{1}{1-e^{-D(\theta)}}\]
  is the optimal coverage.
 \end{thm}
 
 \begin{proof}
  Note that
  \begin{eqnarray}
   P(\text{Correct reconstruction}) &=& P(\text{All parity estimates
   are correct}) \nonumber\\
   &\geq& \left(1- \left(\frac{1}{n}\right)^{c(p)(1-e^{-D(\theta)})}\right)^{n-1}, \nonumber \\
   &\rightarrow& 1, \text{ if } c(p)> \frac{1}{1-e^{-D(\theta)}}.
  \end{eqnarray}
  Further, we have 
  \begin{align}
   P(\text{Correct reconstruction}) &\leq \Bigg(1- \frac{1}{C_1\log^2 n} \left(\frac{1}{n}\right)^{c(p)(1-e^{-D(\theta)})}   
    + \frac{1}{n C_2\log^2 n} + \frac{1}{n C_3\log n} \Bigg)^{n-1}, \nonumber \\
   &\rightarrow 0, \text{ if } c(p)< \frac{1}{1-e^{-D(\theta)}},
  \end{align}
  where $C_1,C_2,C_3$ are positive constants as derived in
  \eqref{eq:converse1}.
 \end{proof}


 \section{Noisy reads covering non-adjacent SNPs} \label{sec:non_adjacent}
 
 Next, we consider the case where mate-pairs cover more than just the
adjacent SNPs. In particular, we will assume that each mate-pair 
covers $S_i,S_{i+W}$, where $i$ is uniform as before, and $W$ is a 
random integer between $1$ and $w$, independently chosen for each 
read. $W$ represents the separation between the two reads of a mate-pair read
measured in terms of number of SNPs between them.

We will consider three cases, in order of increasing
complexity:
 \subsection{$W$ is either $1$ or $2$ with equal probability}
 Let us first consider the case where we have
 observations of adjacent parities and skip-$1$ parities, $i.e.$
 parities $S_i+S_{i+1}$ and $S_i+S_{i+2}$ respectively. Let $N_i^{(1)}$
 be the number of noisy observations of $S_i+S_{i+1}$ and 
 $N_i^{(2)}$ be the number of noisy observations of $S_i+S_{i+2}$.
 Similarly to the previous case, each read consists of
 a uniformly chosen $i$, paired with $i+1$ or $i+2$ with probability
 $1/2$.  Therefore with a total of $\text{Poiss}(c(p)n\log n)$ reads, we have
 $N_i^{(1)}, N_i^{(2)} \sim \text{Poiss}(\frac{c(p)}{2} \log n)$.
 
 Let $Z_{i,j}^{(1)}$ be the $j$th noisy observation of $S_i+S_{i+1}$, 
 and similarly $Z_{i,j}^{(2)}$ be the
 $j$th noisy observation of $S_i+S_{i+2}$
 $(i,i+2)$. That is,
 \begin{eqnarray}
  Z_{i,j}^{(1)}= L_i \oplus \epsilon_{i,j}^{(1)}; \ 1 \leq i \leq n-1, \ 1 \leq j \leq N_i^{(1)}, \\
  Z_{i,j}^{(2)}= L_i \oplus L_{i+1} \oplus \epsilon_{i,j}^{(2)}, \ 1 \leq i \leq n-2, \ 1 \leq j \leq N_i^{(2)}.
 \end{eqnarray}
 
%

\begin{figure}[ht]
\centering
\begin{subfigure}[b]{0.6\linewidth}
\includegraphics[width=0.6\linewidth]{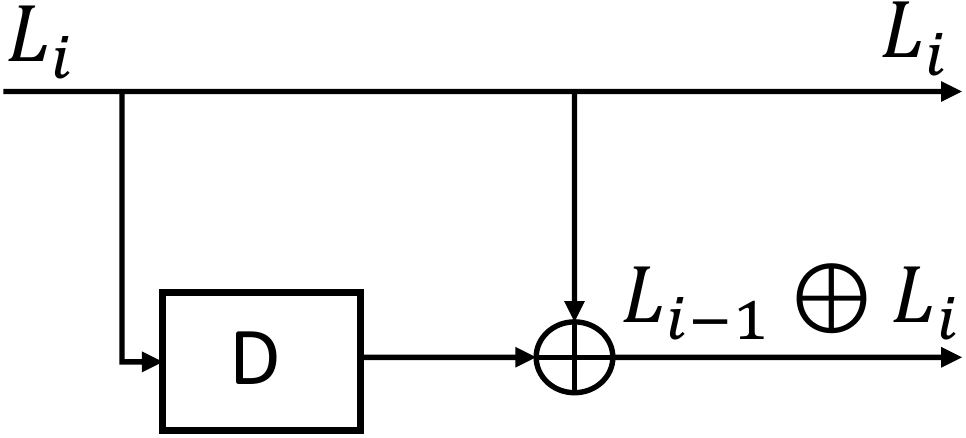}
\label{fig:skip_1}
\end{subfigure}
\quad
\begin{subfigure}[b]{0.3\linewidth}
\includegraphics[width=0.6\linewidth]{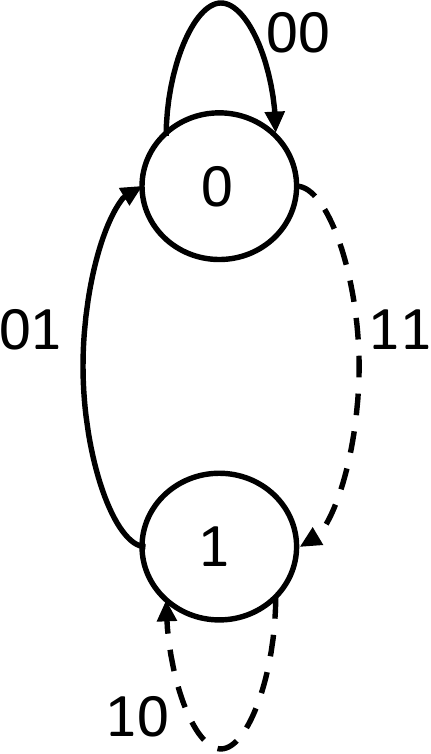}
\label{fig:skip_1_state}
\end{subfigure}
\caption{\small The convolutional code corresponding to the case where we see adjacent and skip-$1$ parities (on the left) and its state diagram (on the right). The transitions corresponding to input $0$ is shown by solid lines and those corresponding to input $1$ are shown by dashed lines. Each output here is seen correctly Poiss$((1-\theta)\frac{ c(p)}{2}\log n)$ times and incorrectly Poiss$(\theta \frac{c(p)}{2}\log n)$ times. \vspace{0.1 in}}
\label{fig:conv}
\end{figure}

 Note that the assembly problem is now equivalent to decoding the $L_i$'s from noisy observations
 of the $L_i$'s and $(L_i+L_{i+1})$'s.  Observe that this is equivalent to decoding 
 a rate $\frac{1}{2}$ convolutional code whose polynomial generator matrix is
 given by (Figure~\ref{fig:conv})
 \begin{equation}
  g(z) = \left[ \begin{array}{cc} 1 & 1+z \end{array}\right].
 \end{equation}
 Each output bit of the code is repeated  Poiss$(
 \frac{c(p)}{2}\log n)$ times, and passed through a binary
 symmetric channel (BSC) with cross over probability $\theta$.
 Unlike coding for communication, however, here one cannot initiate/terminate the code
 as desired, since the values of SNPs are given, and hence we can not $0$ pad as in the communication setting.   ML decoding
 can be done using the 
 Viterbi algorithm, with a trellis that has $2$ states
 (Figure~\ref{fig:conv}). The points on each edge in the trellis corresponds to the number of observations which agree with the output of the transition between the states.

 \begin{thm}\label{thm:skip1} In the setting above,
 \[c(p) = \frac{2/3}{1-e^{-D(\theta)}} \] is the optimal coverage.  \end{thm}
  
 \begin{proof} 
 We will see that the error performance is determined by the {\it free
 distance} of the convolutional code, which is the minimum weight of any non-zero codeword.
 First note that this particular code is not catastrophic that is, no
 infinite-weight input produces a finite-weight output, as
 the GCD of the generator polynomials is $1$.  Thus to calculate the
 free distance, we can restrict ourselves to input message polynomials
 with finite degree.  In particular, note that any input bits
 with weight more than $1$ will have give rise to an output
 of weight at least $4$.  Any input with weight $1$ gives us output
 of weight $3$. Thus we conclude that the free distance of the code is $3$.
 
  Following the analysis in \cite{Vit}, consider the condition where
  the parties are $L_1=L_2=\cdots=L_{n-1}=0$. Next note that any state transition apart from the $0 \rightarrow 0$ transition adds an output of weight at least $1$. Thus we note that the amount of weight accumulated by any path that diverged from the all $0$ path $\ell$ stages ago accumulates weight at least $\ell$. 
  
  Let $R_n(p)$ be the average number of reads that cover a position (that is an output of a convolutional code). We have that,
  \begin{align*}
   R_n(p) = \frac{c(p)}{2} \log n
  \end{align*}

  We further note that since the reads are all independently
  corrupted, the number of reads supporting the all zero path as
  compared with some weight $d$ path depends only on the positions
  where the all zero path and the second path differ. We note that the
  reads supporting the weight $d$ path at these points of difference
  from the all zero path is $X_d \sim \text{Poiss}(d \theta
  R_n(p))$ (because the sum of independent Poisson random
  variables is a Poisson with rate equal to the sum of rates.). The
  points accumulated by the all zero path on these points is $X_0 \sim
  \text{Poiss}(d (1-\theta)R_n(p))$. Further we have that $X_0{\perp\!\!\!\perp} X_d $, for all $d>1$. Thus we have that
  the probability that a path of weight $d$ is preferred over the all
  zero path with probability $P(X_d > X_0)$, which can be bounded as
  (See Appendix \ref{app:poisson_race})


  \begin{align}
   P(X_d > X_0) &\geq \frac{1}{C_1 \log^2 n}\left( \frac{1}{n} \right)^{d\gamma(p)}- \frac{C_2}{n^{2c(p)}  \log n} - \frac{C_3}{n^{2c(p)}  \log^2 n},\\
   P(X_d \geq X_0) &\leq \left( \frac{1}{n} \right)^{d\gamma(p)}, 
  \end{align}
where $\gamma(p) =\frac{c(p)(1-e^{-D(\theta)})}{2}$.

  We first note that the all $0$ path should not be killed by the weight $3$ path that diverged from the all zero path $3$ stages before the current stage. We note that there are $\frac{n}{3}$ disjoint events. Thus we have that,
  \begin{eqnarray}
   P(\text{error}) &\geq& \frac{n}{3} P(X_3 > X_0), \nonumber \\
   &=& \Theta \left(\frac{n}{\log^2 n} \left(\frac{1}{n}\right)^{\frac{3c(p)(1-e^{-D(\theta)})}{2}}\right), \nonumber \\
   &\rightarrow& 1, \text { if } c(p) < \frac{2/3}{1-e^{-D(\theta)}},
  \end{eqnarray}
 thus giving us that 
 \begin{equation}
  c(p) \geq \frac{2/3}{1-e^{-D(\theta)}}
 \end{equation}
 is necessary.
  
  Next, we prove sufficiency. We say that the all-zero path is killed at stage $i$ 
  if some other path is preferred to the all $0$ path at the  $0$ node in the trellis 
  at stage $i$. At the last stage, this includes, the event that the path terminating at the $1$ node of the trellis 
  has accumulated more weight than the path terminating at the $0$ node. 
  
  Let $p_i$
  be  the probability of the all-zero path being killed at stage $i$ of
  the trellis.  We will bound $p_1+\dotsc,+p_n$. 

  First consider the case where the all zero path is killed in the
  first $3$ stages.  As there are only $8$ paths in the trellis at this stage, the probability of this occurring is 
  \begin{equation}
   p_1+p_2+p_3 \leq 8 \frac{1}{n^{\gamma(p)}} \rightarrow 0 \text{ as } n \rightarrow \infty.
  \end{equation}

  Next consider the probability of the all zero codeword being killed
  at the last stage. Note that each stage has at most $2$ surviving paths, and any path of weight $w$ has to have diverged from the all-zero path at most $w$ stages ago, there are at most $2^w$ paths of weight $w$ competing with the all zero codeword at any stage. Thus union bounding this probability is upper bounded as,
  \begin{equation}
   p_n\leq \sum_{i=1}^{\infty} \frac{2^i}{n^{\gamma(p)i}} \leq \frac{4}{n^{\gamma(p)}},
  \end{equation}
  for large enough $n$, which also goes to $0$ as $n\rightarrow \infty$.

  Finally consider the case of the all zero codeword being killed in
  the mid section of the trellis.  Note that after the first $3$ stages any path that did not diverge from the all zero path has hamming weight at least $3$. Further note that from the definition of free distance, any path that diverged from the all zero path and is competing with the all zero path has weight at least $3$. Thus the probability of the all zero codeword being killed at any stage is upper bounded by,
  \begin{equation}
   p_i \leq \sum_{i=3}^{\infty} \frac{2^i}{n^{\gamma(p)i}} = \Theta(\frac{1}{n^{3\gamma(p)}})
  \end{equation}
  
  The total probability of error thus is asymptotically
  \begin{eqnarray}
   P(\text{error}) \leq \lim_{n\rightarrow \infty }  \sum_{i=1}^n p_i \rightarrow 0, \ \text{ if } c(p) > \frac{2/3}{1-e^{D(\theta)}}.
  \end{eqnarray}
  \end{proof}
  
  \begin{note}
   This result actually shows that asymptotically no algorithm can perfectly reconstruct  even in this non-Bayesian setting when $c(p)<  \frac{2/3}{1-e^{-D(\theta)}}$ for all assignments of $L_1,\cdots L_{n-1}$. 
   
   To see this we first note that by the symmetry in the problem and the ML algorithm, the probability of error for every assignment of $L_1,\cdots L_{n-1}$ is the same for the ML algorithm. Thus if any algorithm $\mathcal{A}$ succeeds to perfectly reconstruct when $c(p)<  \frac{2/3}{1-e^{-D(\theta)}}$, then the probability of error of this algorithm on every assignment of $L_1,\cdots L_{n-1}$, would have to be asymptotically lower than the probability of error of the ML algorithm on that assignment.
   
   Further, we note that if each $L_i\sim$Bern$(0.5)$, then the ML algorithm would be the MAP algorithm, with the minimum probability of error. However this leads to a contradiction, because the probability of error of $\mathcal{A}$ in this Bayesian setting would then be asymptotically lower than that of the MAP algorithm.
  \end{note}

  \subsection{ $W$ is uniform over $1, \dotsc ,w$ for $w\ge 3$.}
 We can observe adjacent, skip$-1$, skip$-2, \cdots$, skip$-(w-1)$, parities, each being equally likely. Here let $X^{(\ell)}_i\sim \text{Poiss}(\frac{c(p)}{w} \log n)$ be the number of observation of skip$-m$ parities, $(i,i+m+1)$. With notation as before we have that, each observed parity can be represented as,
  \begin{equation}
   Z_{i,j}^{(\ell)}=\bigoplus_{k=0}^{\ell} L_{i+k} \oplus \epsilon_{i,j}^{(\ell)},
  \end{equation}
  for $ \ 1 \leq i \leq n-\ell, \ 1 \leq j \leq X_i^{(\ell)}, 1 \le \ell \le w-1,$ where $\epsilon_{i,j}^{(\ell)}$ are Bern$(\theta)$ random variables independent of each other.
  
  We note that these can be represented by  a rate $\frac{1}{w}$ convolutional code with polynomial generator matrix
  \begin{equation}
   g_w(z) = \left[\begin{array}{c c c c}
           1 & 1 +z & \cdots & 1 + z + \cdots +z^{w-1}
          \end{array}\right].
  \end{equation}
  We further note that, the trellis corresponding to this code will have $2^{w-1}$ states.

  \begin{lem} \label{lem:free_dist}
   The free distance of the code whose polynomial generator matrix is given by $g_w(z)$ is $2w$, for $w\geq 3$.
  \end{lem} 
  \begin{proof}
 We note that, this code is not catastrophic, as the GCD of the polynomials is $1$. Further, when $w\geq3$, we note that any monomial input will have output weight $\frac{w(w+1)}{2}$, which for $w\geq3$ is greater than $2w$. Further, we note that for any input with more than $1$ monomial, we will each term in the polynomial generator matrix will have outputs of weight at least $2$. Thus we have that the free distance $\geq 2w$. We note that the input $1+z$ gives us an output of weight $2w$. Hence, we have that the free distance of this code is $2w$.  
  \end{proof}

  Following the exact same procedure as before, we have that,
  \begin{thm}
  For the setting above, when $w\ge 3$,
   \[c(p) = \frac{w}{d_{\text{free}}(1-e^{-D(\theta)})} = \frac{1/2}{1-e^{-D(\theta)}},\]
   is the optimal coverage, where  $d_{\text{free}}$ is the free distance of the convolutional code, with polynomial generator matrix $g_w(z)$.
  \end{thm}
  

\subsection{ $W$ is non-uniform}
  Next we consider the case where, all parities are not observed in
  the same proportions.  In particular, the separation between the two reads in a mate-pair measured in terms of number of SNPs between them is a random variable
  $W$, taking integral values between $1$ and $w$, with probabilities $p_1,\dotsc,p_w$.
  That is, the number of observations of skip$-(\ell-1)$ parities $(i,i+\ell)$ is given by $X_i^{(\ell)} \sim \text{Poiss}({c(p)} p_i \log n)$, where $\sum_{i=1}^w p_i=1$.   
  
 We further assume that the GCD of the generator polynomials corresponding to the support of $p_1,\cdots, p_w$ is $1$. Let $\mathbf{p}=(p_1,p_2,\cdots,p_w)$.
 
 To tackle this case we first for a general rate $\frac{1}{w}$ convolutional code $\mathcal{C}$, with $w$ output streams, we define the averaged distance of two  a codewords $v_1, v_2 \in \mathcal{C}$, as
  \begin{equation}
   \tilde{d}_{\mathbf{p}}(v_1,v_2) = \sum_{i=1}^{w} p_i \text{wt}(v_1^{(i)}-v_2^{(i)}),
  \end{equation}
  where $v_j^{(i)}$ is the codeword $v_j$ in the $i$-th stream, and  $\text{wt}(v_1^{(i)}-v_2^{(i)})$ is the hamming distance between $v_1$ and $v_2$ in the $i$-th stream. Further let $\tilde{d}_{\mathbf{p}}(v):= \tilde{d}_{\mathbf{p}}(v,\mathbf{0})$, be referred to as the averaged weight of a codeword $v \in \mathcal{C}$

  For any convolutional code $\mathcal{C}$, define the averaged free distance,
  \begin{equation}
   \tilde{d}_{\text{free}}(\mathbf{p})= \min_{v_1,v_2 \in \mathcal{C}} \tilde{d}_{\mathbf{p}}(v_1,v_2)= \min_{v \in \mathcal{C}- \{\mathbf{0}\}} \tilde{d}_{\mathbf{p}}(v),
  \end{equation}
  where the second equality follows from the linearity of $\mathcal{C}$, and $\mathbf{0}$ is the all zero codeword. Henceforth, by abuse of notation, we shall represent $\tilde{d}_{\text{free}}(\mathbf{p})$ by $\tilde{d}_{\text{free}}$.

  \begin{lem}
   If the GCD of the generator polynomials corresponding to the support of $\mathbf{p}=(p_1,p_2, \cdots, p_w)$ is $1$ ($i.e.$ the code is not catastrophic), for the family of codes under consideration (with polynomial generator matrix $g_w(z), w \geq 2$)), we have that,
 \begin{equation*}
  \tilde{d}_{\text{free}}= \min (\sum_{i=1}^w i p_i, 2 ) =  \min
  \{\mathbb{E}(W), 2\}, 
 \end{equation*}
 where 
 $\mathbb{E}(W) =\sum_{i=1}^w i p_i$.
  \end{lem}

  \begin{proof}
  As the GCD of the polynomials corresponding to the support of $\mathbf{p}$ is $1$, we have that the no input message polynomial of infinite  weight can have an output codeword of finite averaged weight. 
  
  Further we note that for every input of more than $1$ monomial the output on every stream will have at least weight $2$, thus giving us that the averaged weight of any such message symbol is at least $2$. Further, we note that the input monomial $1+z$ will lead to a codeword with weight exactly $2$ on each stream, giving us a codeword of averaged weight $2$.  
  
  Next, we note that any input of $1$ monomial will give rise a codeword of  weight $i$ on the $i$-th stream and hence gives a codeword with averaged weight $ \sum_{i=1}^w i p_i$. 
  
  Thus we have that 
  \[\tilde{d}_{\text{free}}= \min (\sum_{i=1}^w i p_i, 2 ).\]
  \end{proof}

%
 
 \begin{thm} \label{thm:gen_case}
  In this case, 
  \begin{equation}
   c(p) = \frac{1}{\tilde{d}_{\text{free}}(1-e^{-D(\theta)})}=\frac{1}{\min \{ \mathbb{E}(W), 2 \} (1-e^{-D(\theta)})},
  \end{equation}
  is the optimal coverage.
  \end{thm} 
 

 \begin{proof}
  This is relegated to Appendix \ref{app:gen_case}.
 \end{proof}
 \vspace{0.1 in}
 
 \begin{note}

   On the ring of polynomials over the binary field $\mathbb{F}_2[z]$, let,
   \begin{equation}
    v_{r}(z)= 1+z+z^2+\cdots+z^{r-1}.
   \end{equation}
   It is easy to see that if $r|s$, then $v_{r}(z)|v_{s}(z)$. For $1<r\leq s \leq 50$, one can check that if $r \nmid s$, then $v_r(z)\nmid v_s(z)$. We note that all polynomials in the polynomial generator matrix of the code considered here are of the form $v_r(z)$.

   Thus we have that, for $1 \le w\le 50$, the  GCD of generator polynomials corresponding to the support of $p_1,\cdots, p_w$ is not equal to 1, only occurs  when Support$(p_1,p_2,\cdots,p_w)$ is a set of the form $\{ 1 < j \leq w: j=ki, \text{ for some } k \in \mathbb{N}\}$, for some integer $i>1$. This corresponds to the case when our observations have a cyclic structure, in which case reconstruction is impossible.
    \end{note}
   
    \section{Non uniform SNP coverage}
    
 In all the above results, we have assumed that the number of reads covering every SNP position, with SNPs of its neighbourhood are identically distributed. In general, the genomic distance between adjacent SNP pairs
 will not be constant. There may be SNPs which do not have many SNPs in their near them on the genome, and hence may be covered by much fewer reads that cover multiple SNPs. We consider this in the setting of Section \ref{sec:adjacent}, where only adjacent SNPs are covered by reads, as an illustration of how this can be handled. We first start with a simple example.
 
  \begin{example}
 Suppose we only observe adjacent SNPs.
 If $\alpha$ fraction of adjacent SNPs were covered with probability $\frac{t}{n-1}$, $t<1$ and the rest were covered with equal probability, then 
 $ c(p)= \frac{1/t}{1-e^{-D(\theta)}}$ is the optimal coverage. This can be shown using calculations identical to those of \ref{sec:adjacent}.
 \end{example}

 Next, we consider a more general case.
 Suppose then that we only observe adjacent SNPs, and  that the probability of each read
 covering SNPs $(i,i+1)$ is not $\tfrac1{n-1}$ but instead 
 $\frac{q_i^{(n)}}{n-1}$ (where $\sum_{i=1}^{n-1}q_i^{(n)}=n-1$), and hence the number of observations of the
 $i$th parity is Poiss$(q_i^{(n)} c(p) \log(n))$.  Further assume that
 $q_i^{(n)}\in[C_1,C_2]$ for some constants $0<C_1\le1\le C_2$.   Fix
 $\delta >0$ and let $T_{\delta}=\lceil \frac{C_2 - C_1}{\delta} \rceil$.
 
 For any $\delta>0$, $1 \leq \ell \leq T_{\delta}$, let,
 \begin{eqnarray}
  S^{(n)}_{\ell,\delta} &:=& \{ i: C_1 + (\ell-1) \delta < q_i^{(n)} \leq C_1 \ell \delta\}, \\
  \epsilon_{\delta,\ell} &:=&  \inf_{\epsilon \in (0,1]}\{ |S^{(n)}_{\ell,\delta}| \in O(n^{\epsilon}) \},
 \end{eqnarray}
 and define $m$ and $k$ as
 \begin{equation}
  m:= \inf_{0 < \delta \leq C_2-C_1} \max_{1 \leq \ell \leq T_\delta} \frac{\epsilon_{\delta,\ell}}{C_1+(\ell-1)\delta}.
 \end{equation}
 and
 \begin{equation}
 k:= \sup_{0 < \delta \leq C_2-C_1} \max_{1 \leq \ell \leq T_\delta} \frac{\epsilon_{\delta,\ell}}{C_1+\ell\delta}.
\end{equation}
 
Then we have that 
\begin{eqnarray}
 P(\text{Perfect Recovery}) &=& \prod_{i=1}^{n-1}\left(1- \left(\frac{1}{n}\right)^{c(p)q_i^{(n)}(1-e^{-D(\theta)})}\right), \nonumber \\
 &\geq& \prod_{\ell=1}^{T_{\delta}}\left(1- \left(\frac{1}{n}\right)^{c(p) (C_1+ (\ell-1)\delta)(1-e^{-D(\theta)})}\right)^{|S_{\ell,\delta}^{(n)}|},\nonumber \\
 &\rightarrow& 1, \text{ if }c(p)> \frac{m}{1-e^{-D(0.5||2p(1-p))}}, \text{ as } n \rightarrow \infty.
\end{eqnarray}

This gives us that
 \begin{equation}
  c(p)> \frac{m}{1-e^{-D(\theta)}}
 \end{equation}
 is sufficient for perfect reconstruction, and  similar calculations to those of Section \ref{sec:adjacent} give us that,  
 
 \begin{equation}
  c(p)\geq \frac{k}{1-e^{-D(\theta)}}
 \end{equation}
 is necessary.

{
 \bibliographystyle{plain}
\bibliography{phasing}
}

  \appendix
  \subsection{Poisson Races} \label{app:poisson_race}
  If $X \sim$ Poiss$(\lambda)$ and $Y \sim $ Poiss$(\mu)$, $ \mu> \lambda$, $X {\perp\!\!\!\perp} Y$ then from a Chernoff bound, we have that,
 \begin{equation}
  P(X\leq Y) \leq \exp(- (\sqrt{\lambda} -\sqrt{\mu} )^2).
 \end{equation}
 
  Further noting, that $X+Y \sim$ Poiss $(\lambda +\mu)$. And $X|X+Y \sim$ Bin$(X+Y, \frac{\lambda}{\lambda+\mu})$. One can show that,
  \begin{equation}
   P(X>Y) > \frac{\exp(-(\sqrt{\mu} -\sqrt{\lambda})^2  )}{(\lambda + \mu)^2} - \frac{e^{-(\lambda + \mu)}}{2\sqrt{\lambda \mu}} - \frac{e^{-(\lambda + \mu)}}{4\lambda \mu},
  \end{equation}
  by noting that $P(X>Y)= P(X> \frac{X+Y}{2})$, and upper bounding that term conditioned on $X+Y=i$, $\forall i \in \mathbb{Z}$.

  \subsection{Proof of Theorem \ref{thm:gen_case}}\label{app:gen_case}

 As in proof of Theorem \ref{thm:skip1}, the probability that a averaged weight $d$ path is preferred to the all zero path is given by (as this is a race between a the number of erroneous reads at positions of difference $\sim$ Poiss$(d\theta c(p)\ln n )$ and the number of correct reads at positions of difference $\sim$ Poiss$(d(1-\theta)c(p) \ln n )$ .),
 \begin{equation*}
  P[\text{averaged weight $d$ path kills all zero}] \le  \left(\frac{1}{n^{c(p) (1- e^{-D(\theta)})}}\right)^{d}.
 \end{equation*}
 
  Let,
 
 \begin{equation*}
  \tau = \min_{\substack{i\in \{1,2, \cdots, w \} \\ p_i > 0}} p_i.
 \end{equation*}
 \begin{equation*}
  \gamma(p)={c(p) (1- e^{-D(\theta)})}
 \end{equation*}

 We thus have that in every $2^{w-1}+1$ stages any path on the trellis that does not visit the all zero stage adds a averaged weight of at least $\tau$ (because code is not catastrophic implies, that we can not get from any state other than the all zero state to itself, without adding any weight, and one has to visit some state twice in $2^{w-1}+1$ stages). Let $\epsilon = \left\lceil{\frac{\tilde{d}_{free}}{\tau}}\right\rceil$. Thus the number of paths in the trellis that have averaged weight less than $\tilde{d}_{free} + \ell \tau$, is at most $2^{w(2^{w-1} +1)(\ell + \epsilon)}$.
 
 As in the proof of Theorem \ref{thm:skip1}, we can show that the probability that the all $0$ path will be killed in the first $2^{w(2^{w-1} +1)\epsilon}$ stages or the last stage go to $0$ as $n\rightarrow \infty$. We now restrict our attention to the all $0$ path being killed in the middle stages, where any codeword competing with the all $0$ codeword has averaged weight at least $\tilde{d}_{free}$.
 
 The probability of a averaged weight $d$ killing the all zero path is less than the probability that a averaged weight $\tilde{d}_{free} + \ell  \tau$ path kills the all zero, where $\ell$ is picked such that $\tilde{d}_{free} + \ell  \tau < d \leq \tilde{d}_{free} + (\ell+1)  \tau$. Thus we have that, the probability of error at any stage here is,
 \begin{align*}
 P(E) &\leq  \left(\frac{2^{w(2^{w-1} +1)\epsilon}}{n^{\gamma(p)}}\right) ^{\tilde{d}_{free}} \left( \sum_{\ell=1}^{\infty}   \left(\frac{2^{w(2^{w-1} +1)}}{n^{\gamma(p)}}\right)^{\ell} \right) \\
 &= \left(\frac{1}{n^{\gamma(p)}}\right) ^{\tilde{d}_{free}} O(1),
 \end{align*}
 for large enough $n$.
 
 Thus, by union bounding the probability of error across all stages as before, we can show that if 
 \begin{equation}
  c(p) > \frac{1}{\min \{ \mathbb{E}(W), 2 \} (1-e^{-D(\theta)})},
 \end{equation}
 then we have perfect reconstruction.
 
 The argument for necessity is essentially the same as that of Theorem \ref{thm:skip1}, with a similar quantization argument.

\end{document}